\documentclass[aps,twocolumn, longbibliography,superscriptaddress,floatfix,]{revtex4-2}

\usepackage{graphicx}
\usepackage{amsmath}
\usepackage{amssymb}
\usepackage{amsthm}
\usepackage{algpseudocode}
\usepackage{algorithm}
\usepackage{braket}
\usepackage{comment}

\usepackage[colorlinks = true]{hyperref}
\usepackage{xcolor}
\definecolor{darkred}  {rgb}{0.5,0,0}
\definecolor{darkblue} {rgb}{0,0,0.5}
\definecolor{darkgreen}{rgb}{0,0.5,0}

\hypersetup{
	urlcolor   = blue,         
	linkcolor  = darkblue,     
	citecolor  = darkgreen,    
	filecolor  = darkred       
}
\usepackage{cleveref}

\newcommand{\be}{\begin{equation}}
	\newcommand{\ee}{\end{equation}}
\newcommand{\ba}{\begin{array}}
	\newcommand{\ea}{\end{array}}
\newcommand{\bea}{\begin{eqnarray}}
	\newcommand{\eea}{\end{eqnarray}}

\newcommand{\ra}{\rangle}
\newcommand{\la}{\langle}

\newcommand{\calC}{{\cal C }}
\newcommand{\calS}{{\cal S }}

\newcommand{\EE}{\mathbb{E}}

\newcommand{\CC}{\mathbb{C}}

\newtheorem{lemma}{Lemma}
\newenvironment{lemmaprime}[1]
{\renewcommand{\thelemma}{\ref{#1}$'$}
	\addtocounter{lemma}{-1}
	\begin{lemma}}
	{\end{lemma}}
\newenvironment{lemmacopy}[1]
{\renewcommand{\thelemma}{\ref{#1}}
	\addtocounter{lemma}{-1}
	\begin{lemma}}
	{\end{lemma}}

\newtheorem{theorem}{Theorem}

\newcommand{\ketbra}[1]{\ket{#1}\bra{#1}}
\newcommand{\defeq}{\stackrel{\text{def}}{=}}

\begin{document}

	\title{How Much Entanglement Is Needed for Quantum Error Correction?}
	
	\author{Sergey Bravyi}
	\affiliation{IBM Quantum, IBM T. J. Watson Research Center, Yorktown Heights, New York 10598, USA}
	
	\author{Dongjin Lee}
	\affiliation{Perimeter Institute for Theoretical Physics, Waterloo, Ontario N2L 2Y5, Canada}
	\affiliation{Department of Physics and Astronomy, University of Waterloo, Waterloo, Ontario N2L 3G1, Canada}

	\author{Zhi Li}
	\affiliation{Perimeter Institute for Theoretical Physics, Waterloo, Ontario N2L 2Y5, Canada}
        \affiliation{National Research Council Canada, Waterloo, Ontario N2L 3W8, Canada}

	\author{Beni Yoshida}
	\affiliation{Perimeter Institute for Theoretical Physics, Waterloo, Ontario N2L 2Y5, Canada}

	\begin{abstract}

		It is commonly believed that logical states of quantum error-correcting codes have to be highly entangled such that codes capable of correcting more errors require more entanglement to encode a qubit. 
        Here, we show that the validity of this belief depends on the specific code and the choice of entanglement measure.
        To this end, we characterize a tradeoff between the code distance $d$ quantifying the number of correctable errors, and the geometric entanglement measure of logical states quantifying their maximal overlap with product states or more general ``topologically trivial" states. The maximum overlap is shown to be exponentially small in $d$ for three families of codes: (1) low-density parity check codes with commuting check operators, (2) stabilizer codes, and (3) codes with a constant encoding rate. 
		Equivalently, the geometric entanglement of any logical state of these codes grows at least linearly with $d$. On the opposite side, we also show that this distance-entanglement tradeoff does not hold in general. For any constant $d$ and $k$ (number of logical qubits), we show there exists a family of codes such that the geometric entanglement of some logical states approaches zero in the limit of large code length.
	\end{abstract}
	
	\maketitle
	
	Quantum error correction \cite{shor1995scheme} and entanglement theory are vibrant research fields that are often closely interlinked.
	A deep connection between the two fields stems from the hallmark feature of entangled states known as local indistinguishability \cite{knill1997theory,bennett1996mixed}---two entangled states may be perfectly distinguishable globally while looking identical for any local observer who can only examine few-qubit subsystems.
	Accordingly, a logical qubit encoded into entangled states that span multiple physical qubits and satisfy the local indistinguishability property becomes protected from all few-qubit errors enabling error correction.
	
	The exchange of ideas between these fields has led to several breakthroughs.
	For example, the study of many-body entangled states with a topological quantum order~\cite{wen1990ground} has led to the development of fault-tolerant architectures based on the surface code~\cite{kitaev2003fault,dennis2002topological}.
	Conversely, coding theory constructions such as holographic quantum codes~\cite{pastawski2015holographic,almheiri2015bulk} provided a fruitful perspective on the connection between geometry and entanglement in quantum gravity.
	More recently, one of the most significant open questions about many-body entanglement---the existence of local Hamiltonians without low-energy topologically trivial states~\cite{freedman2013quantum}---has been resolved by employing the construction of good low density parity check (LDPC) codes~\cite{anshu2023nlts,panteleev2022asymptotically}.

	Given the profound role of entanglement in quantum coding theory,
	a natural question is {\em how much entanglement is needed for quantum error correction?}
	Here, we address this question by exploring distance-entanglement tradeoffs for quantum codes.
	Recall that a quantum code encoding $k$ logical qubits into $n$ physical qubits is simply a linear subspace $\calC$ of dimension $2^k$ embedded into the $2^n$-dimensional Hilbert space describing $n$ physical qubits. 
	Any normalized vector $|\psi\ra\in \calC$ represents a logical (encoded)  state.
	The error-correcting capability of a code is usually measured by its {\em distance}.
 A code is said to have distance at least $d$ if logical states are indistinguishable on any subset of less than $d$ physical qubits. 
	Namely, for any fixed subset of less than $d$ qubits, the reduced density matrices are the same of all different logical states. 
	A distance-$d$ code can correct any error affecting less than $d/2$ physical qubits~\cite{knill1997theory,bennett1996mixed,nielsen2010quantum}. 
	It is easy to check that logical states of nontrivial codes ($k\ge 1$ and $d\ge 2$) must have {\em some} entanglement~\footnote{Indeed, if $d\ge 2$ then all logical states must have the same single-qubit marginals.
		However, any state sharing the same single-qubit marginals with an unentangled (product) state must be equal to the unentangled state itself.
		Hence, a code with $k\ge 1$ and $d\ge 2$ cannot have an unentangled logical state.}. But what is the minimum amount of entanglement required for a distance-$d$ code?

    We provide a fairly complete answer to this question.	  
    Using a generalized ``geometric entanglement measure" (GEM) to quantify entanglement, we establish a distance-entanglement tradeoff for three pivotal families of quantum error-correcting codes: (1) LDPC codes, (2) stabilizer codes, and (3)  codes with a constant encoding rate $k/n$.
    For codes in these families, our results show that the geometric entanglement must grow at least linearly with the code distance (or even the code length).
    Our bounds have practical implications for quantum code design: 
    realizing codes with the geometric entanglement sublinear in the code distance may be more challenging as they would require high-weight non-Pauli check operators and introduce significant encoding overhead.
    
    Conversely, we demonstrate that this distance-entanglement tradeoff is not universal.
    For any given code parameters $k$ and $d$, there exists a family of codes where the geometric entanglement of certain logical states vanishes as the code length $n$ grows.

    To quantify the entanglement of $n$-qubit states, our starting point is the geometric entanglement measure~\cite{shimony1995degree,barnum2001monotones,wei2003geometric,orus2014geometric},
    which measures the maximum overlap between a given $n$-qubit state and product states.
	Unfortunately, this measure can be ``spoofed" by states exhibiting only short-range entanglement.
	For example, a tensor product of $n/2$ Einstein-Podolsky-Rosen pairs scores high on the geometric entanglement even though
	this state can be disentangled by a single layer of controlled-\textsc{not} gates.
	To detect genuine many-body entanglement we generalize the geometric entanglement beyond product states.
	For any $h\ge 1$ define a depth-$h$ quantum circuit acting on $n$ qubits as a composition of $h$ layers of two-qubit gates such that 
	gates within each layer are nonoverlapping.
	We allow two-qubit gates acting on any pair of qubits (all-to-all qubit connectivity).
	Depth-$0$ circuits are defined as products of single-qubit gates. 
	Define the depth-$h$ geometric entanglement measure of an $n$-qubit state $|\psi\ra$ as
	\be
	\label{E_h}
	E_h(\psi) =  -\max_{U\, : \, \mathrm{depth}(U)= h}\;   \log_2{|\la \psi| U|0^n\ra|^2},
	\ee
	where the maximum is taken over all $n$-qubit depth-$h$ circuits.
	By definition, $E_h(\psi)\in [0,n]$ and $E_h(\psi)=0$ if and only if $|\psi\ra$ can be prepared by a depth-$h$ circuit.
	We shall be interested in the regime when $h=O(1)$ is a constant independent of $n$.
	Then $E_h(\psi)$ quantifies the maximum overlap between $|\psi\ra$ and ``topologically trivial" states~\cite{freedman2013quantum}, i.e. 
	states prepared by a constant-depth circuit starting from product states. 
	The standard GEM ~\cite{shimony1995degree,barnum2001monotones,wei2003geometric,orus2014geometric} coincides with $E_0(\psi)$.
	We shall also consider Clifford GEM $E_h^{C}(\psi)$ defined by Eq.~(\ref{E_h}) with a restriction that $U$ contains only Clifford gates.
	
	\section{Main results}
	First, consider quantum LDPC (qLDPC) codes.
	Such codes are often used to model systems with topological quantum order, see e.g.~\cite{kitaev2003fault,bravyi2010topological}.
	The code space $\calC$ of a qLDPC code with $n$ physical qubits  is defined as
	\be
	\calC= \{ |\psi\ra \in (\CC^2)^{\otimes n} \, : \, \Pi_a|\psi\ra =|\psi\ra~~ \forall \; a=1,\ldots,m\}
	\ee
	where $\Pi_1,\ldots,\Pi_m$ are pairwise commuting Hermitian projectors.
	Equivalently, $\calC$ is the ground subspace of a gapped Hamiltonian $H=-\sum_{a=1}^m \Pi_a$. 
	The code is said to have {\em sparsity} $s$ if
	each projector $\Pi_a$
	acts nontrivially on at most $s$ qubits and each qubit participates in at most $s$ projectors.
	The LDPC condition demands that $s$ is a constant independent of $n$.

	Our first result is the distance-entanglement tradeoff for qLDPC codes. 
	\begin{theorem}
		\label{thm:LDPC}
		For any $d>s^42^{5h}$, any $s$-sparse, distance-$d$ qLDPC code $\calC$, and any logical state $|\psi\ra\in \calC$, one has
		\be
		\label{bound_LDPC}
		E_h(\psi) \ge \alpha d,
		\ee
		where $\alpha>0$ is a constant that depends only on $s$ and $h$. 
	\end{theorem}
	As a consequence, the overlap between any logical state and any topologically trivial state is exponentially small in $d$.
	Moreover, since the bound Eq.~(\ref{bound_LDPC}) does not depend on the number of physical qubits $n$,
	the constant depth circuit that creates a topologically trivial state may act on an arbitrary number of ancillary qubits
	initialized in $|0\ra$. 
	The bound Eq.~(\ref{bound_LDPC}) confirms the belief that high-distance codes capable of correcting more errors require more entanglement to encode
	a qubit. 
	We emphasize that this bound is extremely general as it applies to {\em all} qLDPC codes.
	The linear scaling with $d$ in Eq.~(\ref{bound_LDPC}) cannot be improved  since
	$E_h(\psi)\le n$ for any state $|\psi\ra$ and there exist qLDPC codes with the distance $d\sim n$~\cite{panteleev2022asymptotically}.
	Although the bound Eq.~(\ref{bound_LDPC}) is easy to state, its proof is rather nontrivial.
	It combines the characterization of entanglement present in
	logical states of qLDPC codes stated as the disentangling lemma in~\cite{bravyi2010tradeoffs} and the Hamming weight concentration bound for shallow peaked quantum circuits (slightly improved upon \cite{bravyi2023classical}).

	Our second result is the distance-entanglement tradeoff for stabilizer codes~\cite{gottesman1998theory}.
	The code space $\calC$ of a stabilizer code with $n$ physical qubits is defined as
	\be
	\label{stabilizer_code}
	\calC=\{  |\psi\ra \in (\CC^2)^{\otimes n} \, : \, S_a|\psi\ra =|\psi\ra~~ \forall \; a=1,\ldots,m\}
	\ee
	where $S_a\in \{\pm 1 \} \cdot  \{I,X,Y,Z\}^{\otimes n}$ are pairwise commuting Pauli operators
	and $m=n-k$. 
	We show the following. 
	\begin{theorem}
		\label{thm:stabilizer}
		For any stabilizer code $\calC$ and any logical state $|\psi\ra\in \calC$ one has
		\be
		\label{bound_stabilizer}
		E_0(\psi) \ge d-1.
		\ee
	\end{theorem}
	Equivalently, the overlap between any logical state and any $n$-qubit product state
	is at most $2^{1-d}$.
	Moreover, this bound is
	tight as it can be saturated exactly by Shor's code~\cite{shor1995scheme}; see Supplemental Material \cite{supp} for details. \nocite{janson2004large,movassagh2020constructing}
	The proof of Eq.~(\ref{bound_stabilizer}), which relies on the cleaning lemma~\cite{bravyi2009no},
	is pleasingly simple. 
	Although we were not able to prove a lower bound $\Omega(d)$ for the depth-$h$ GEM,
	a simple corollary of Eq.~(\ref{bound_stabilizer}) is 
	\be
	\label{bound_stabilizer_Clifford}
	E_h^C(\psi) \ge \frac{d}{2^h} - 1.
	\ee
	In other words, the overlap between any logical state  and any $n$-qubit state that
	can be prepared by a depth-$h$ Clifford circuit is at most $2^{1-d/2^h}$, exponentially small in $d$ for any constant depth $h$.

	While the above results depend only on $d$ (but not $k$), our third result contains $k$-dependent lower bounds, which is particularly useful for codes with constant rates.
	Let $H(x)=-x\log_2 x-(1-x)\log_2(1-x)$, and $H^{-1}(x)$ be its inverse function restricting to $x\in[0,\frac{1}{2}]$.
	\begin{theorem}\label{thm:overlapgeneralcode}
		\begin{enumerate}
			\item[(i)] For any code $\calC$ and any logical state $|\psi\ra\in \calC$ one has
			\begin{equation}\label{eq:knbound}
				E_h(\psi)\geq\Big(\frac{d}{2^h}-1\Big)H^{-1}\Big(\frac{k}{n}\Big).
			\end{equation}
			\item[(ii)] For a qLDPC code such that $d>s^42^{5h}$, 
			\begin{equation}\label{eq:knLDPC}
				E_h(\psi)\geq \beta nH^{-1}\Big(\frac{k}{n}\Big).
			\end{equation}
			where $\beta=\Theta(s^{-4}2^{-4h})>0$ is a constant that depends only on $s$ and $h$. 
		\end{enumerate}
	\end{theorem}
	In particular, if the code has constant rate $k/n=\Theta(1)$, then for general codes, we have
	\begin{equation}\label{eq:knboundasmp}
		E_h(\psi)=\Omega\Big(\frac{d}{2^h}\Big);
	\end{equation}
	for qLDPC code such that $d>s^42^{5h}$, we have
	\begin{equation}
		E_h(\psi)=\Omega\Big(\frac{n}{2^{4h}}\Big).
	\end{equation}
	
	Equation (\ref{eq:knbound}) is proved by combining an entropic argument inspired by \cite{bravyi2010tradeoffs,anshu2020circuit} and a general observation about making measurements on a code.
	Equation (\ref{eq:knLDPC}) relies on similar ideas toward Eq.~(\ref{bound_LDPC}).

	Finally, one might expect that a similar distance-entanglement tradeoff 
	holds for all quantum codes. Surprisingly, we show that this is not the case.
	To this end, we consider permutation invariant codes~\cite{pollatsek2004permutationally,ouyang2014permutation}.
	Logical states of such codes are invariant under any permutation of physical qubits. 
	We construct a family of distance-$2$ permutation invariant codes $\calC_n$ encoding $k=1$ logical qubit into $n$ physical qubits
	and a family of logical states $|\psi_n\ra\in \calC_n$ with an asymptotically vanishing GEM,
	\be
	\label{PI_code}
	E_0(\psi_n) =O(\frac{1}{n}).
	\ee
	Hence the overlap between $|\psi_n\ra$ and some $n$-qubit product state approaches $1$ in the limit $n\to \infty$.
	By concatenating the code $\calC_n$ with itself several times we obtain a family of codes with arbitrarily large $d$ and $k$ such that a logical state obeys $E_0(\psi_n)\to 0$ as $n\to\infty$.
	This example shows that some logical states 
	of high-distance codes may have very little entanglement.

	\noindent
	{\em Related works.} 
	Circuit lower bounds for low-energy states of quantum code Hamiltonians have been studied before by many authors~\cite{aharonov2018quantum,freedman2013quantum,eldar2017local,anshu2023nlts,anshu2020circuit}.
	Upper bounds on the overlap between logical states and low-depth states
	similar to ours were reported by Anshu and Nirkhe~\cite{anshu2020circuit}
	who showed $E_h(\psi)=\Omega(\frac{k^2d}{n^2 s2^{4h}\ln^4{(ds)}})$ for quantum LDPC codes (Lemma 13 therein).
	Compared with Eq.~(\ref{bound_LDPC}), this bound depends on the code length $n$ and is sublinear in $d$ (even up to logarithm) unless the code has constant rate $k/n$, 
	in which case the LDPC condition is actually not necessary; see Eq.~(\ref{eq:knboundasmp}).
	Reference~\cite{anshu2020circuit} also showed $E_h(\psi)=\Omega(d^2/n2^{2h})$ for any code (Lemma 14 therein).
	In particular, $E_h(\psi)=\Omega(d)$ for codes with linear distance and $h=O(1)$, which complements our results.
	Reference~\cite{botero2007scaling} used GEM to probe quantum phase transitions in several 1D spin chain models; see also~\cite{orus2008universal,orus2008geometric,wei2010entanglement}.
	Reference~\cite{orus2014geometric} studied GEM for some examples of topological quantum codes.
	In an accompanying work \cite{li2024much}, we study the GEM $E_h$ for many-body systems that support emergent anyons and fermions.

	\section{proofs}
	We proceed to the proof of the distance-entanglement tradeoffs stated above.

	\subsection{Stabilizer codes}
	We begin with Theorem \ref{thm:stabilizer} since its proof is the simplest one.
	Suppose $\calC$ is a stabilizer code with distance $d$ and $|\psi\ra\in \calC$ is a logical state.
	By definition, $2^{-E_0(\psi)} =\la \psi|\rho|\psi\ra$ for some 
	product state $\rho=\rho_1\otimes \cdots \otimes \rho_n$.
	Partition $n$ qubits into two disjoint registers $A,B$ 
	and let $\rho_A=\mathrm{Tr}_B \rho$.
	Since $\rho$ is a product state, one has $\rho\le \rho_A \otimes I_B$.
	Accordingly,
	\be
	\label{stabilizer_eq1}
	2^{-E_0(\psi)} = \la \psi|\rho|\psi\ra\le \la \psi|\rho_A\otimes I_B|\psi\ra
	=\mathrm{Tr}(\rho_A \eta_A),
	\ee
	where $\eta=|\psi\ra\la \psi|$ and  $\eta_A=\mathrm{Tr}_B \eta$.
	Suppose $|A|<d$. The local indistinguishability property implies that the reduced
	density matrix $\eta_A$
	is the same for all logical states. Thus one can compute $\eta_A$ by pretending that $\eta$
	is the maximally mixed logical state $\eta=(1/2^k)\Pi_\calC$, where $\Pi_\calC$ is the 
	projector onto $\calC$. It is well known~\cite{nielsen2010quantum} that 
	\begin{equation}
		\Pi_\calC = \frac1{2^{n-k}} \sum_{S\in \calS} S,
	\end{equation}
	where $\calS=\la S_1,\ldots,S_m\ra$ is the stabilizer group of $\calC$.
	Taking the partial trace over $B$ gives
	\be
	\label{stabilizer_eq2}
	\eta_A = \frac1{2^k} \mathrm{Tr}_B \Pi_\calC =\frac1{2^{|A|}} \sum_{S\in \calS(A)} S, 
	\ee
	where $\calS(A)$ is the subgroup of $\calS$ consisting of all stabilizers whose support is a subset of $A$.
	We shall need the following lemma.
	\begin{lemma}
		\label{lemma:1}
		For any distance-$d$ stabilizer code, there exists a subset of qubits $A$ such that $|A|\geq d$ and the only stabilizer supported on $A$ is the identity operator.
	\end{lemma}
	Applying Eqs.~(\ref{stabilizer_eq1}) and (\ref{stabilizer_eq2}) to a subset of qubits $A$
	such that $\calS(A)=\{I\}$ and $|A|=d-1$ gives
	$\eta_A=2^{1-d}I$ and
	\be
	2^{-E_0(\psi)}\le 2^{1-d} \mathrm{Tr}(\rho_A) = 2^{1-d},
	\ee
	which proves the distance-entanglement tradeoff (\ref{bound_stabilizer}).
	
	Furthermore, suppose $U$ is a depth-$h$ Clifford circuit, then $U\calC$ is also a stabilizer code and the standard light cone argument shows that its distance is at least
	$d/2^h$.
	Applying Eq.~(\ref{bound_stabilizer}) to the stabilizer code $U\cdot \calC$
	proves Eq.~(\ref{bound_stabilizer_Clifford}).

	\begin{proof}[\bf Proof of Lemma~\ref{lemma:1}]
		We shall use induction in $\ell=|A|$.
		The base of induction is $\ell=0$ in which case $A=\emptyset$ and $\calS(A)=\{I\}$ by definition. 
		Consider the induction step. Let $A$ be a subset of $\ell$ qubits such that $\calS(A)=\{I\}$.
		If $\ell=d$ we are done. Otherwise
		$|A|=\ell<d$.
		The cleaning lemma~\cite{bravyi2009no} asserts that for any logical Pauli operator $\overline{P}$ 
		and any subset of qubits $A$ with $|A|<d$
		there exists an equivalent logical Pauli operator $\overline{Q} \in \overline{P}\cdot \calS$
		such that $\overline{Q}$ acts trivially on $A$.
		Thus one can choose an anticommuting pair of logical Pauli operators $\overline{X}$ and $\overline{Z}$ that both act trivially on $A$. 
		Then there must exist a qubit $i\notin A$ such that $\overline{X}$ and $\overline{Z}$ locally anticommute on $i$.
		Assume without loss of generality (WLOG) that $\overline{X}$ and $\overline{Z}$ act on the $i$th qubit as Pauli $X$ and $Z$
		respectively
		(otherwise, perform a local Clifford change of basis on the $i$th qubit). 
		Consider a set $A'=A\cup \{i\}$ of size $\ell+1$. We claim that $\calS(A')=\{I\}$.
		Indeed, let  $S\in \calS$ be a stabilizer supported on $A'$. 
		Let $S_i\in \{I,X,Y,Z\}$ be the restriction of $S$ onto the $i$th qubit.
		Since $S$ commutes with all logical operators and $S$ overlaps with 
		$\overline{X}$ and $\overline{Z}$ {\em only} on the $i$th qubit, one infers that
		$S_i$ commutes with both Pauli $X$ and $Z$. This is only possible if $S_i=I$.
		Thus $S\in \calS(A)=\{I\}$, that is, $S=I$ and $\calS(A')=\{I\}$.
		This completes the induction step.
	\end{proof}
	
	\subsection{Quantum LDPC codes}
	Let us prove Theorem \ref{thm:LDPC}.
	It suffices to consider an arbitrary logical state $|\psi\ra\in \calC$ and prove that
	\be
	\label{LDPC_eq1}
	-\log_2{|\la \psi|0^n\ra|^2} \ge d\cdot g(s)
	\ee
	for some function $g(s)>0$.
	Indeed, since a local change of basis on each qubit does not affect the distance and sparsity of $\calC$, we can assume WLOG that $E_0(\psi)$ is achieved by $-\log_2{|\la \psi|0^n\ra|^2}$.
	Moreover, suppose $U$ is a depth-$h$ circuit, then $U\cdot \calC$ is a qLDPC code defined by commuting projectors $U\Pi_a U^\dag$. 
	The standard light cone argument shows that $U\cdot \calC$ has sparsity at most $s 2^h$ and distance at least $d /2^{h}$.
	Hence, $E_0(\psi)\ge d\cdot g(s)$ implies $E_h(\psi)\ge d\cdot g(s2^h)/2^h$. 
	
	The key is to consider the probability distribution of bit strings corresponding to measuring $\ket{\psi}$ in the computational basis:
	\begin{equation}
		\mathrm{Pr_\psi}(x)=|\braket{x|\psi}|^2,~~x\in\{0,1\}^n.
	\end{equation}
	The following lemma characterizes such probability distribution.
	\begin{lemma}
		\label{lemma:x}
		Suppose $d>s^4$. Then:
		\begin{enumerate}
			\item[(i)] each bit of $x$ is independent of all but at most $K$ other bits;
			\item[(ii)] $\EE_\psi(|x|)\leq -(K+1)\ln{|\la \psi|0^n\ra|^2}$;
			\item[(iii)] there exists a function $c(s)$ such that $ \mathrm{Pr_\psi}[|x|\geq t]\leq \exp{\left[ \EE_\psi(|x|)-t\right]}$ for any $t\ge c(s) \EE_\psi(|x|)$.
		\end{enumerate}
	\end{lemma}
	\noindent
	Here, $K=s^2+s^4$; $|x|=\sum_{i=1}^n x_i$ is the Hamming weight of $x$; $\ln$ stands for the natural logarithm.
	In the following, we omit the subscripts in $\mathrm{Pr_\psi}$ and $\EE_\psi$.
	
	Item (i) follows from the disentangling lemma \cite{bravyi2010tradeoffs}, 
	which implies that each qubit is disentangled from the rest after tracing out a finite number of qubits.
	Items (ii) and (iii) are consequences of (i). They state that if the probability distribution is ``peaked" at $0^n$, i.e., $|\braket{\psi|0^n}|^2$ is large, then most of the probability is supported in a small Hamming ball around $0^n$, with an exponential tail.
	Full proofs can be found in Supplemental Material \cite{supp}.

	Now we derive Eq.~(\ref{LDPC_eq1}) from Lemma~\ref{lemma:x}.
	We will apply the inclusion-exclusion principle to evaluate $|\la \psi|0^n\ra|^2$ from marginal distributions.
	We denote $R = -(K+1)\ln{|\la \psi|0^n\ra|^2}$.
	If $d<c(s)R$, or equivalently
	\be
	\label{LDPC_eq1_part1}
	-\ln{|\la \psi|0^n\ra|^2} > \frac{d}{(K+1)c(s)}, 
	\ee
	then we already have the desired form (\ref{LDPC_eq1}).
	
	Now assume $d\ge c(s)R$.
	Define
	\begin{equation}
		S_i=\sum_{1\leq  p_1<\cdots <p_i\leq n}\Pr[x_{p_1}=\cdots=x_{p_i}=1],
	\end{equation}
	for each integer $i\in [1,n]$.
	Then we have (recall the derivation of Bonferroni's inequality)
		\begin{align}\label{LDPC_eq2}
			&(-1)^{d-1}\left[ \Pr[|x|>0] - \sum_{i=1}^{d-1}(-1)^{i-1}S_i \right] \nonumber\\
			=&\sum_{t=d}^{n}\binom{t-1}{d-1}\Pr[|x|=t] 
			= \sum_{t=d}^{n}\binom{t-2}{d-2}\Pr[|x|\geq t]\nonumber\\
			\leq&\sum_{t=d}^{\infty}\binom{t-2}{d-2}e^{R-t}
			= e^{-1+R} (e-1)^{1-d}. 
		\end{align}
	Here the inequality uses Lemma~\ref{lemma:x} (ii) and (iii).
	The  condition $t\ge c(s) \EE(|x|)$ in (iii) is satisfied since $d\ge c(s)R$.
	The last equality in Eq.~(\ref{LDPC_eq2}) follows from binomial expansion (with a negative exponent); $e\equiv \exp{(1)}$. 
	
	Crucially, $S_1,\ldots,S_{d-1}$ depend solely on reduced density matrices
	of $|\psi\ra$ on subsets of less than $d$ qubits.
	Consequently, these quantities are identical for any logical states $|\psi\ra$ due to the local indistinguishability.
	Equation~(\ref{LDPC_eq2}) hence shows the overlap $|\la \psi|0^n\ra|^2 = 1-\Pr[|x|>0]$ is approximately the same for any logical states, up to the correction $e^{-1+R} (e-1)^{1-d}$. 
	On the other hand, since $\dim\calC\geq 2$, there always exists a logical state $|\psi'\ra\in \calC$ such that $\la \psi'|0^n\ra=0$.
	Therefore 
	\be
	\label{LDPC_eq4}
	|\la \psi|0^n\ra|^2\le e^{-1+R} (e-1)^{1-d}
	\ee
	for any logical state $\ket{\psi}\in\calC$.
	Equivalently,
	\be
	\label{LDPC_eq1_part2}
	-\ln{|\la \psi|0^n\ra|^2} \geq \frac{1+(d-1)\ln{(e-1)}}{K+2},
	\ee
	which is in the desired form (\ref{LDPC_eq1}).
	Since either Eq.~(\ref{LDPC_eq1_part1})
	or Eq.~(\ref{LDPC_eq1_part2}) must hold,
	we have proved Eq.~(\ref{LDPC_eq1})  with
	$
	g(s)=\frac1{\ln{2}} \min{\left[ \frac1{(K+1)c(s)}, \frac{\ln(e-1)}{K+2} \right]}.
	$	
	
	\subsection{Constant rate codes}
	
	Now let us prove Theorem \ref{thm:overlapgeneralcode}.
	Let us again consider $|\braket{0^n|\psi}|^2$ without loss of generality.
	
	To prove item (i), we will inductively measure $(d-1)$ qubits, one at a time. 
	The following lemma guarantees the existence of a qubit that is ``mixed enough."
	\begin{lemma}\label{lemma:codeentropybound}
		For any code state, there exists a qubit $q$ that the von Neumann entropy satisfies $S_{q}\geq k/n$. 
	\end{lemma}
	\noindent
	Assuming it is the first qubit WLOG, Lemma \ref{lemma:codeentropybound} guarantees that we can expand $\ket{\psi}$ as
	\begin{equation}\label{eq:decomp2}
		\ket{\psi}=\lambda_0\ket{0}\ket{\psi_0}+\lambda_1\ket{1}\ket{\psi_1}
	\end{equation}
	such that $\lambda_0, \lambda_1\neq 0$. 
	Moreover, it also follows that $H(|\lambda_0|^2)\geq S_{1}\geq k/n$, hence $|\lambda_0|^2\leq 1-H^{-1}(k/n)$. 
	
	The state $\ket{\psi_0}$ is a postmeasurement state after measuring the first qubit in the computational basis.
	As a general principle, it must be a logical state of a (different) code:
	\begin{lemma}\label{lemma:postselectcode}
		Measuring one qubit in a distance $d$ code, the postmeasurement state (regardless of the measurement outcome) is a logical state of a code with $k'=k$ and $d'\geq d-1$.
	\end{lemma}
	
	Both lemma \ref{lemma:codeentropybound} and \ref{lemma:postselectcode} are direct consequences of the local indistinguishability of quantum codes, and are proved in Supplemental Material \cite{supp}. 
	With them in mind, let us repeat the procedure at least $(d-1)$ times until the code distance might decrease to 1. Each time we can upper bound the corresponding amplitude using Lemma \ref{lemma:codeentropybound}.
	The overlap is then bounded by the product
	\begin{equation}\label{eq:overlapboundgeneral}
		|\langle 0^n | \psi \rangle |^2  \leq  \prod_{i=0}^{d-2} \Big[1-H^{-1}\Big(\frac{k}{n-i}\Big)\Big].
	\end{equation}
	Equation (\ref{eq:knbound}) then follows from $-\ln(1-x)>x$ and that a depth-$h$ circuit at most decreases the code distance by a factor $2^h$.
	
	For item (ii), we use Lemma \ref{lemma:x}. 
	Because of the linearly of the expectation value and the local indistinguishability of quantum codes,
	$\EE_\psi(|x|)$ is the same for $\forall\ket{\psi}\in \calC$.
	In particular, if choosing $a$ so that $H(\frac{a}{n})=\frac{k}{n}$,
	then $\dim\calC$ is large enough so that there always exists $\ket{\psi}\in\calC$ such that $\braket{\psi|x}=0$ for all $x$ whose Hamming weight $|x|\leq a$:
	\begin{equation}\label{eq:nofcondition}
		\sum_{i=0}^{a}\binom{n}{i}< 2^{nH(\frac{a}{n})}\leq 2^k.
	\end{equation}
	For such $\ket{\psi}$, it is clear that $\EE_\psi(|x|)> a$.
	If follows from Lemma \ref{lemma:x}(ii) that 
	\begin{equation}
		-\ln{|\la \psi|0^n\ra|^2}>\frac{a}{K+1}=\frac{n}{K+1}H^{-1}\Big(\frac{k}{n}\Big).
	\end{equation}

	\subsection{Permutation Invariant codes}
	
	Consider a code $\calC$ with $k=1$ logical qubit and $n\ge 4$ physical qubits such that 
	the logical states encoding $|0\ra$ and $|1\ra$ are 
	\be
	\label{logical_basis0}
	|\psi_0\ra = \sqrt{1-\frac{2}{n}} |0^n\ra + \sqrt{\frac{2}{n}} |1^n\ra 
	\ee
	and
	\be
	\label{logical_basis1}
	|\psi_1\ra = \sqrt{\frac{2}{n(n-1)}} \sum_{x\in \{0,1\}^n\, : \, |x|=2} \; \; |x\ra
	\ee
	respectively. 
	In other words, $\calC$ is the two-dimensional subspace spanned by $|\psi_0\ra$ and $|\psi_1\ra$.
	The code $\calC$ is an example of a permutation invariant code~\cite{pollatsek2004permutationally,ouyang2014permutation}
	since the logical states $|\psi_i\ra$ are invariant under any permutation of $n$ qubits. 
	A simple calculation reveals that $\calC$ has distance $d=2$,  see Supplemental Material~\cite{supp}. 
	By definition of GEM,
	\be\label{eq:PIcodeGEM}
	2^{-E_0(\psi_0)} \ge |\la \psi_0|0^n\ra|^2  = 1-\frac2n.
	\ee
	Hence $E_0(\psi_0)\le -\log_2[1-(2/n)]\le 4/n$.
    We note that it also implies vanishing Renyi-$\alpha$ entanglement entropy for any bipartite cut for any $\alpha>1$.
	By concatenating the code $\calC$ with itself several times we obtain a family of 
	distance-$d$ codes
	with an arbitrarily large $d$ and $k$ such that their logical states
	have GEM approaching zero in the limit $n\to \infty$, see Supplemental Material \cite{supp} for details.

	\section{Conclusions}
	We have established distance-entanglement tradeoffs for three broad families of quantum codes: LDPC codes
	with commuting check operators, stabilizer codes, and constant rate codes.
	Logical states of such codes are shown to be highly entangled
	such that the geometric entanglement measure grows at least linearly with the code distance.
	This highlights the role of entanglement as a resource enabling quantum error correction based on
	the above code families.
	At the same time, we show that there exist families of high-distance codes such that some logical states may have very little entanglement.

	\begin{acknowledgments} 
		S.B. thanks Chinmay Nirkhe for helpful discussions. 
		Research at Perimeter Institute is supported in part by the Government of Canada through the Department of Innovation, Science and Economic Development and by the Province of Ontario through the Ministry of Colleges and Universities.
	\end{acknowledgments}
	\bibliographystyle{unsrt}
	\bibliography{QECbib}

\begin{thebibliography}{10}

\bibitem{shor1995scheme}
Peter~W Shor.
\newblock Scheme for reducing decoherence in quantum computer memory.
\newblock {\em Physical review A}, 52(4):R2493, 1995.

\bibitem{knill1997theory}
Emanuel Knill and Raymond Laflamme.
\newblock Theory of quantum error-correcting codes.
\newblock {\em Physical Review A}, 55(2):900, 1997.

\bibitem{bennett1996mixed}
Charles~H Bennett, David~P DiVincenzo, John~A Smolin, and William~K Wootters.
\newblock Mixed-state entanglement and quantum error correction.
\newblock {\em Physical Review A}, 54(5):3824, 1996.

\bibitem{wen1990ground}
Xiao-Gang Wen and Qian Niu.
\newblock Ground-state degeneracy of the fractional quantum {H}all states in the presence of a random potential and on high-genus {R}iemann surfaces.
\newblock {\em Physical Review B}, 41(13):9377, 1990.

\bibitem{kitaev2003fault}
A~Yu Kitaev.
\newblock Fault-tolerant quantum computation by anyons.
\newblock {\em Annals of physics}, 303(1):2--30, 2003.

\bibitem{dennis2002topological}
Eric Dennis, Alexei Kitaev, Andrew Landahl, and John Preskill.
\newblock Topological quantum memory.
\newblock {\em Journal of Mathematical Physics}, 43(9):4452--4505, 2002.

\bibitem{pastawski2015holographic}
Fernando Pastawski, Beni Yoshida, Daniel Harlow, and John Preskill.
\newblock Holographic quantum error-correcting codes: Toy models for the bulk/boundary correspondence.
\newblock {\em Journal of High Energy Physics}, 2015(6):1--55, 2015.

\bibitem{almheiri2015bulk}
Ahmed Almheiri, Xi~Dong, and Daniel Harlow.
\newblock Bulk locality and quantum error correction in ads/cft.
\newblock {\em Journal of High Energy Physics}, 2015(4):1--34, 2015.

\bibitem{freedman2013quantum}
Michael~H. Freedman and Matthew~B. Hastings.
\newblock Quantum systems on non-k-hyperfinite complexes: a generalization of classical statistical mechanics on expander graphs.
\newblock {\em Quantum Info. Comput.}, 14(1–2):144–180, 2014.

\bibitem{anshu2023nlts}
Anurag Anshu, Nikolas~P Breuckmann, and Chinmay Nirkhe.
\newblock {NLTS} hamiltonians from good quantum codes.
\newblock In {\em Proceedings of the 55th Annual ACM Symposium on Theory of Computing}, pages 1090--1096, 2023.

\bibitem{panteleev2022asymptotically}
Pavel Panteleev and Gleb Kalachev.
\newblock Asymptotically good quantum and locally testable classical {LDPC} codes.
\newblock In {\em Proceedings of the 54th Annual ACM SIGACT Symposium on Theory of Computing}, pages 375--388, 2022.

\bibitem{nielsen2010quantum}
Michael~A Nielsen and Isaac~L Chuang.
\newblock {\em Quantum computation and quantum information}.
\newblock Cambridge university press, 2010.

\bibitem{Note1}
Indeed, if $d\ge 2$ then all logical states must have the same single-qubit marginals. However, any state sharing the same single-qubit marginals with an unentangled (product) state must be equal to the unentangled state itself. Hence, a code with $k\ge 1$ and $d\ge 2$ cannot have an unentangled logical state.

\bibitem{shimony1995degree}
Abner Shimony.
\newblock Degree of entanglement.
\newblock {\em Annals of the New York Academy of Sciences}, 755(1):675--679, 1995.

\bibitem{barnum2001monotones}
Howard Barnum and Noah Linden.
\newblock Monotones and invariants for multi-particle quantumstates.
\newblock {\em Journal of Physics A: Mathematical and General}, 34(35):6787, 2001.

\bibitem{wei2003geometric}
Tzu-Chieh Wei and Paul~M Goldbart.
\newblock Geometric measure of entanglement and applications to bipartite and multipartite quantum states.
\newblock {\em Physical Review A}, 68(4):042307, 2003.

\bibitem{orus2014geometric}
Rom{\'a}n Or{\'u}s, Tzu-Chieh Wei, Oliver Buerschaper, and Maarten Van~den Nest.
\newblock Geometric entanglement in topologically ordered states.
\newblock {\em New Journal of Physics}, 16(1):013015, 2014.

\bibitem{bravyi2010topological}
Sergey Bravyi, Matthew~B Hastings, and Spyridon Michalakis.
\newblock Topological quantum order: stability under local perturbations.
\newblock {\em Journal of mathematical physics}, 51(9), 2010.

\bibitem{bravyi2010tradeoffs}
Sergey Bravyi, David Poulin, and Barbara Terhal.
\newblock Tradeoffs for reliable quantum information storage in 2d systems.
\newblock {\em Physical review letters}, 104(5):050503, 2010.

\bibitem{bravyi2023classical}
Sergey Bravyi, David Gosset, and Yinchen Liu.
\newblock Classical simulation of peaked shallow quantum circuits.
\newblock In {\em Proceedings of the 56th Annual ACM Symposium on Theory of Computing}, STOC 2024, page 561–572, New York, NY, USA, 2024. Association for Computing Machinery.

\bibitem{gottesman1998theory}
Daniel Gottesman.
\newblock Theory of fault-tolerant quantum computation.
\newblock {\em Physical Review A}, 57(1):127, 1998.

\bibitem{supp}
See Supplemental Material, which includes Refs.~\cite{janson2004large,movassagh2020constructing}, for detailed proofs of the theorems in the main text, as well as an additional theorem stating that the number of terms in the expansion of any code state is at least $2^d$.

\bibitem{janson2004large}
Svante Janson.
\newblock Large deviations for sums of partly dependent random variables.
\newblock {\em Random Structures \& Algorithms}, 24(3):234--248, 2004.

\bibitem{movassagh2020constructing}
Ramis Movassagh and Yingkai Ouyang.
\newblock Constructing quantum codes from any classical code and their embedding in ground space of local hamiltonians.
\newblock {\em Quantum}, 8:1541, 2024.

\bibitem{bravyi2009no}
Sergey Bravyi and Barbara Terhal.
\newblock A no-go theorem for a two-dimensional self-correcting quantum memory based on stabilizer codes.
\newblock {\em New Journal of Physics}, 11(4):043029, 2009.

\bibitem{anshu2020circuit}
Anurag Anshu and Chinmay Nirkhe.
\newblock {Circuit Lower Bounds for Low-Energy States of Quantum Code Hamiltonians}.
\newblock In {\em 13th Innovations in Theoretical Computer Science Conference (ITCS 2022)}, volume 215 of {\em Leibniz International Proceedings in Informatics (LIPIcs)}, pages 6:1--6:22, Dagstuhl, Germany, 2022. Schloss Dagstuhl -- Leibniz-Zentrum f{\"u}r Informatik.

\bibitem{pollatsek2004permutationally}
Harriet Pollatsek and Mary~Beth Ruskai.
\newblock Permutationally invariant codes for quantum error correction.
\newblock {\em Linear algebra and its applications}, 392:255--288, 2004.

\bibitem{ouyang2014permutation}
Yingkai Ouyang.
\newblock Permutation-invariant quantum codes.
\newblock {\em Physical Review A}, 90(6):062317, 2014.

\bibitem{aharonov2018quantum}
Dorit Aharonov and Yonathan Touati.
\newblock Quantum circuit depth lower bounds for homological codes.
\newblock {\em arXiv preprint arXiv:1810.03912}, 2018.

\bibitem{eldar2017local}
Lior Eldar and Aram~W Harrow.
\newblock Local hamiltonians whose ground states are hard to approximate.
\newblock In {\em 2017 IEEE 58th annual symposium on foundations of computer science (FOCS)}, pages 427--438. IEEE, 2017.

\bibitem{botero2007scaling}
Alonso Botero and Benni Reznik.
\newblock Scaling and universality of multipartite entanglement at criticality.
\newblock {\em arXiv preprint arXiv:0708.3391}, 2007.

\bibitem{orus2008universal}
Rom{\'a}n Or{\'u}s.
\newblock Universal geometric entanglement close to quantum phase transitions.
\newblock {\em Physical review letters}, 100(13):130502, 2008.

\bibitem{orus2008geometric}
Rom{\'a}n Or{\'u}s.
\newblock Geometric entanglement in a one-dimensional valence-bond solid state.
\newblock {\em Physical Review A}, 78(6):062332, 2008.

\bibitem{wei2010entanglement}
Tzu-Chieh Wei.
\newblock Entanglement under the renormalization-group transformations on quantum states and in quantum phase transitions.
\newblock {\em Physical Review A}, 81(6):062313, 2010.

\bibitem{li2024much}
Zhi Li, Dongjin Lee, and Beni Yoshida.
\newblock How much entanglement is needed for topological codes and mixed states with anomalous symmetry?
\newblock {\em Physical Review X}, 15(2):021090, 2025.

\end{thebibliography}
	
\appendix
	\section{Shor's code}
	\label{app:Shor}
	
	Distance-$d$ Shor's code~\cite{shor1995scheme}
	is a stabilizer code defined on a two-dimensional grid of qubits of 
	size $d\times d$.  It encodes one logical qubit into $n=d^2$ physical qubits.
	The code has Pauli stabilizers $Z_{i,j}Z_{i+1,j}$ acting on pairs of adjacent qubits located in the same column and $\prod_{i=1}^d X_{i,j}X_{i,j+1}$ acting on pairs of adjacent columns. 
	Logical Pauli operators can be chosen as $X_L=\prod_{i=1}^d X_{i,1}$ along a column and $Z_L=\prod_{j=1}^d Z_{1,j}$ along a row.
	
	It is well known and can be easily checked that the code space is spanned by
	\begin{equation}
		\begin{aligned}
			\ket{\psi_0}&=\Big(\frac{|0^d\ra+|1^d\ra}{\sqrt{2}}\Big)^{\otimes d},\\
			\ket{\psi_1}&=\Big(\frac{|0^d\ra-|1^d\ra}{\sqrt{2}}\Big)^{\otimes d},
		\end{aligned}
	\end{equation}
	where each $(|0^d\ra\pm|1^d\ra)/\sqrt{2}$ is a GHZ state associated with a column of the grid.
	Then $|\psi_+\ra =(\ket{\psi_0}+\ket{\psi_1})/\sqrt{2}$ is a normalized logical state and 
	$E_0(\psi_+) \le -\log_2{|\la \psi_+|0^n\ra|^2}= d-1$ which matches the lower bound of Theorem~\ref{thm:stabilizer}.

	\section{Bit string distribution of qLDPC codes}\label{app:qLDPCproof}
	In this appendix, we prove Lemma~\ref{lemma:x} regarding the bit string distributions corresponding to a logical state of a qLDPC code in the computational basis.
	
	\begin{proof}
		{\em (i)}  We claim that any logical state $|\psi\ra\in \calC$
		has zero correlation length.  More precisely, 
		consider any subset of qubits $M\subseteq [n]$ 
		and let $M^c=[n]\setminus M$ be the complement of $M$.
		Let $\partial M$ be the boundary of $M$ defined as the set of 
		qubits covered by supports of projectors $\Pi_a$ that overlap
		with both $M$ and $M^c$.
		The Disentangling Lemma of Ref.~\cite{bravyi2010tradeoffs} implies that 
		if $|M|<d$ and $|M^c \cap \partial M|<d$ then 
		\be
		\label{LDPC_eq5}
		\mathrm{Tr}_{\partial M}  |\psi\ra\la \psi| = \rho_{M\setminus \partial M}  \otimes 
		\rho_{M^c \setminus \partial M},
		\ee
		where $\rho_A$ is  the reduced density matrix of $|\psi\ra\la \psi|$
		describing a subset of qubits $A$.
		Consider any qubit $j$. We choose $M$ as the union of $j$
		and all neighbors of $j$ in the ``interaction graph" determined by the projectors $\Pi_a$.
		Then $j\in M\setminus \partial M$ and Eq.~(\ref{LDPC_eq5}) implies  that measuring $|\psi\ra$ in the standard basis,
		the $j$-th measured bit 
		is independent of all bits in $M^c\setminus \partial M$. The number of remaining bits
		that can possibly be correlated with the $j$-th bit
		is at most $|M\cup \partial M|-1$. 
		A simple calculation shows that  $|M\cup \partial M|\le s^2+s^4$ while conditions of the Disentangling Lemma are satisfied whenever $d>s^4$.
		
		\noindent
		{\em (ii)} The proof is based on~\cite{bravyi2023classical}. Define a dependency graph $G=(V,E)$
		with $n$ vertices such that the $i$-th bit of $x$ can be correlated only with the
		nearest neighbors of $i$ in $G$. We have already shown that the vertex
		degree of $G$ is at most $K$. If $V_I\subset V$ is an independent set of vertices
		then all bits $x_j$ with $j\in V_I$ are independent. Let $m_j=\EE(x_j)$.
		Then 
		\begin{align}
			|\la \psi|0^n\ra|^2 
			& \le \mathrm{Pr}[x_j=0~~\forall j\in V_I] \nonumber\\
			&= \prod_{j\in V_I} \mathrm{Pr}[x_j=0] =  \prod_{j\in V_I} (1-m_j)\nonumber\\
			&\le \exp{(-\sum_{j\in V_I} m_j)}.\label{LDPC_eq6}
		\end{align}
		Using a simple greedy algorithm, one can show that any graph with the maximum vertex degree $K$ and vertex weights $m_j \in [0,1]$
		has an independent set $V_I$ such that
		\be
		\label{LDPC_eq7}
		\sum_{j\in V_I} m_j \ge \frac1{K+1} \sum_{j=1}^n m_j = \frac{\EE{(|x|)}}{K+1}.
		\ee
		Combining Eqs.~(\ref{LDPC_eq6},\ref{LDPC_eq7}) gives the desired lower bound
		on $\EE{(|x|)}$.

		\noindent
		{\em (iii)} 
		Theorem 2.3 of \cite{janson2004large} gives the following concentration bound for a sum of partly dependent random variables $x_i$ and all $t\ge 0$:
		\begin{equation}
			\Pr[|x|\geq \EE{(|x|)}+t]\leq \exp{\left[ -\frac{S}{K+1}\phi\left(\frac{4t}{5S}\right)\right]}.
		\end{equation}
		Here $\phi(x)=(1+x)\ln(1+x)-x$, and $S=\sum_i \text{Var}(x_i)\leq \EE{(|x|)}$. 
		Since the function $\phi(x)$ is superlinear in $x$, there exists $x_0=x_0(s)$ such that $\phi(x)>\frac{5}{4}(K+1)x$ when $x>x_0$. 
		Hence, if $t>\frac{5}{4}x_0\EE{(|x|)}$, then 
		\begin{equation}
			\Pr[|x|\geq \EE{(|x|)}+t]\leq \exp(-t).
		\end{equation}
		Now we replace $t$ with $t-\EE{(|x|)}$ and set $c(s)=\frac{5}{4}x_0(s)+1$. 
		A simple algebra gives $x_0(s)=\exp{[1+(5/4)(K+1)]}$.
		Accordingly, the function $g(s)$ defined in
		Eq.~(\ref{LDPC_eq1}) can be chosen as
		$g(s)=1/(\ln{2} (K+1)c(s))$ with $K=s^2+s^4$. 
	\end{proof}

	\section{Proof of Lemma~\ref{lemma:codeentropybound}}\label{app-proofentropy}
	
	Here we prove Lemma \ref{lemma:codeentropybound} in the main text, copied below for convenience:
	\begin{lemmacopy}{lemma:codeentropybound}
        For any code state $\ket{\psi}$ in a quantum code $\mathcal{C}$ with code distance $d \geq 2$, there exists at least one qubit $q$ such that the reduced state $\rho_q = \operatorname{Tr}_{\bar{q}} \ketbra{\psi}$ has von Neumann entropy satisfying $S_q\geq k/n$. Here $k=\log_2(\dim\cal C)$.
	\end{lemmacopy}
	\begin{proof}
		We denote the von Neumann entropy of the reduced state of $\ketbra{\psi}$ on qubit $i$ as $S_i$.
		Consider the maximally mixed state in the code space $\cal C$, denoted by $\rho$. By definition, $S(\rho)=k$.
		
		Due to the local indistinguishability, $S_i$ can be equally computed by reducing $\rho$ to qubit $i$. 
		The subadditivity of entropy then implies
		\begin{equation}
			k=S(\rho)\leq\sum_{i=1}^n S_{i}.
		\end{equation}
		Hence, there exists at least one qubit $q$ such that $S_q\geq k/n$.
	\end{proof}

	\section{Measuring a Quantum Code}\label{app-measurecode}
	
	In the proof of Theorem \ref{thm:overlapgeneralcode}, we considered an inductive procedure to measure $d-1$ qubits. It is based on the following:
	\begin{lemmaprime}{lemma:postselectcode}
		Given a code with distance $d$, postselecting $m<d$ qubits either (1) annihilates all states, or (2) results in a code with the same dimension and code distance $d'\geq d-m$.
	\end{lemmaprime}
	\begin{proof}
		Denote $\{\ket{\psi_i}\}$ as a orthonormal basis of the code space $\mathcal{C}$; denote $P$ as the postselection operator. By the Knill-Laflamme (KL) condition \cite{knill1997theory}, we have
		\begin{equation}
			\bra{\psi_i}P^\dagger P\ket{\psi_j}=c(P^\dagger P)\delta_{ij}.
		\end{equation}
		If $P$ does not annihilate $\mathcal{C}$, then $c(P^\dagger P)>0$. The space after projection, $\mathcal{C}'$, is spanned by $\{P\ket{\psi_i}\}$, for which an orthonormal basis is $\{P\ket{\psi_i}/\sqrt{c(P^\dagger P)}\}$, hence $\dim \mathcal{C}'=\dim \mathcal{C}$.
		
		For any operator $\mathcal O$ applied on non-postselected $(d-m-1)$ qubits, $P^\dagger\mathcal{O}P$ applies on at most $(d-1)$ qubits. Therefore, by KL condition (applied to $\ket{\psi_i}$),
		\begin{equation}
			\bra{\psi_i}P^\dagger\mathcal{O}P\ket{\psi_j}=c(P^\dagger\mathcal{O}P)\delta_{ij}\defeq c'(\mathcal{O})\delta_{ij}.
		\end{equation}
		Hence, again due to KL condition (applied to $P\ket{\psi_i}$), $\mathcal{C}'$ is a code with code distance at least $d-m$.
	\end{proof}
	As a simple corollary, we prove a lower bound regarding the number of terms in the wavefunction of any logical state. 
	\begin{theorem}\label{thm:expandnumber}
		Expanding any logical state of a distance-$d$ code in any computational basis, the number of terms must be at least $2^{d-1}$.
	\end{theorem}
	\begin{proof} 
		We prove the theorem by induction on $d$. First, $d=2$ case. $\forall \ket{\psi}, \ket{\phi}\in\cal C$, for each qubit $i$, we have $(\ketbra{\psi})_i=(\ketbra{\phi})_i$. If $\ket{\psi}$ is a product state, we must have $\ket{\phi}=\ket{\psi}$. Therefore, in a code ($k\geq 1$), no logical state can be a product state, implying the number of terms must be at least 2.
		
		Assuming we have proved the $(d-1)$ case, let us consider distance-$d$ codes. For any logical state $\ket{\psi}$, using the $d=2$ case, we know there exists a qubit, say the first qubit, that is not pure if we trace out its complement. We can then expand $\ket{\psi}$ as
		\begin{equation}
			\ket{\psi}=\lambda_0\ket{0}\ket{\psi_0}+\lambda_1\ket{1}\ket{\psi_1},
		\end{equation}
		where $\lambda_i\neq 0$.
		Postselecting on $\ket{0}$, we know from Lemma \ref{lemma:postselectcode} that $\ket{\psi_0}$ is a logical state in a code of distance $d'\geq d-1$, hence must have at least $2^{d-2}$ terms by induction. The same applies to $\ket{\psi_1}$. 
		Therefore, the number of terms in $\ket{\psi}$ must be at least $2^{d-2}+2^{d-2}=2^{d-1}$.
	\end{proof}
	
	This bound is also tight, as it can be saturated by Shor's code as shown in Appendix \ref{app:Shor}.

	\section{Concatenation of PI codes}
	\label{app:PI}
	
	Let us check that the permutation invariant code $\calC$ defined by Eqs.~(\ref{logical_basis0},\ref{logical_basis1}) has distance $d=2$.
	Suppose $\sigma \in \{X_i,Y_i,Z_i\}$ is a single-qubit error. 
	The assumption $n\ge 4$ implies $\la \psi_0|\sigma|\psi_1\ra=0$. Thus, it suffices to check that 
	\be
	\label{KLPI}
	\la \psi_0 |\sigma|\psi_0\ra = \la \psi_1|\sigma|\psi_1\ra.
	\ee
	If $\sigma \in \{X_i,Y_i\}$ then both sides of Eq.~(\ref{KLPI}) are zero.
	Suppose $\sigma=Z_i$. We have
	\be
	\la \psi_0|Z_i|\psi_0\ra = (1-\frac{2}{n}) - \frac{2}{n} = 1-\frac{4}{n}
	\ee
	and
	\begin{equation}
		\begin{aligned}
			\la \psi_1|Z_i|\psi_1\ra &= \frac2{n(n-1)}\left( {n-1 \choose 2} - {n-1\choose 1}\right) \\
			&=1-\frac{4}{n}.   
		\end{aligned}
	\end{equation}
	Here the terms $ {n-1 \choose 2}$ and ${n-1\choose 1}$ count $n$-bit strings $x$ with $|x|=2$
	such that $x_i=0$ and $x_i=1$ respectively.
	This proves Eq.~(\ref{KLPI}).
	As a side remark, we note that the code $\calC$ is a special case of a general 
	construction that maps classical  (non-linear) codes to quantum codes proposed in~\cite{movassagh2020constructing}. In particular, the quantum code $\calC$ is constructed from a classical distance-$2$ length-$n$ code whose codewords are
	$n$-bit strings $x$ with $|x|\in \{0,2,n\}$ by applying the framework of~\cite{movassagh2020constructing}.
	
	Next, let us show how to improve the distance of $\calC$ by concatenation.
	We will show that the concatenated code has logical states
	with the geometric entanglement approaching zero in the limit of large $n$. 
	Let $\calC_n$ be the distance-$2$ 
	code with one logical qubit and $n$ physical qubits defined by Eqs.~(\ref{logical_basis0},\ref{logical_basis1}).
	For any integer sequence $n_1,n_2,\ldots,n_\ell$ let
	$\calC$ be the concatenation of codes $\calC_{n_1},\ldots,\calC_{n_\ell}$, where $\calC_{n_1}$ is at the lowest
	and $\calC_{n_\ell}$ is at the highest  level of concatenation. In other words, each consecutive block of $n_1$ physical qubits
	defines a level-$1$ logical qubit encoded by $\calC_{n_1}$, each consecutive block of $n_2$ level-$1$ logical qubits
	defines a level-$2$ logical qubit encoded by $\calC_{n_2}$, etc. 
	The code $\calC$ has one logical qubit, $N_\ell=\prod_{i=1}^\ell n_i$ physical qubits, and distance $d_\ell = 2^\ell$,
	assuming that $n_i\ge 4$ for all $i$.
	Let $F_\ell$ be the overlap between the logical-$0$ state of the level-$\ell$ logical qubit and the product state $|0^{N_\ell}\ra$.
	From Eq.~(\ref{logical_basis0}) one easily gets
	\be
	F_{i} = \left(1-\frac2{n_i}\right) (F_{i-1})^{n_i}
	\ee
	with $i=1,\ldots,\ell$ and $F_0\equiv 1$. 
	Thus
	\be
	F_\ell = \prod_{i=1}^\ell  \left(1-\frac2{n_i}\right)^{n_{i+1}n_{i+2} \cdots n_\ell},
	\ee
	(for $i=\ell$, $n_{i+1}n_{i+2} \cdots n_\ell=1$).
	
	Let $d\ge 2$ be the desired code distance and $\ell =\lceil \log_2{(d)} \rceil$.
	Let us choose $\{n_i\}$ properly to ensure that the geometric entanglement approaches zero. 
	Let $M\ge 1$ be a large integer.
	Set $n_\ell = 2M$ and 
	\be
	\label{n_i}
	n_i =2M n_{i+1}n_{i+2} \cdots n_\ell 
	\ee
	for all $i\in [1,\ell-1]$ so that
	\be
	F_\ell = \prod_{i=1}^\ell \left(1-\frac1{M n_{i+1}n_{i+2} \cdots n_\ell} \right)^{ n_{i+1}n_{i+2} \cdots n_\ell }.
	\ee
	One can easily check that $(1 - 1/(Mx))^x$ is a monotone increasing function of $x$  for $x\ge 1$.
	Thus
	\be
	\label{F_ell}
	F_\ell \ge \left(1-\frac1M\right)^\ell.
	\ee
	Denoting $|\psi_0^\ell\ra$ the logical-$0$ state of the level-$\ell$ logical qubit, we conclude that 
	\be
	E_0(\psi_0^\ell) \le -\log_2(F_\ell)
	= O(\ell/M)=O(2^d/M).
	\ee
	Thus, for any constant $d\ge 2$ there exists a family of codes (by increasing $M$) such that $E_0(\psi_0^\ell)\to 0$.
	
	Solving Eq.~(\ref{n_i}) gives $n_i=(2M)^{2^{\ell-i}}$, thus the number of physical qubits in the resulting distance-$d$ code $\calC$ equals
	\be
	N_\ell = \prod_{i=1}^\ell n_i = (2M)^{2^\ell-1}<(2M)^{2d}.
	\ee
	
	To construct a code with $k$ logical qubits, we simply take $k$ independent copies of the concatenated code. 
	Equivalently, we concatenate it further with the trivial $[[k,k,1]]$ code.
	An example of a logical state is $|\psi\ra=|\psi_0^\ell\ra^{\otimes k}$.
	The geometric entanglement $E_0(\psi)\le O(k2^d/M)$, which can also be made arbitrarily small by choosing a large enough $M$.

	Finally, let us point out that some distance-$d$ PI codes exhibit GEM scaling at most logarithmically with $d$ (without concatenation).
	Indeed, let $d\ge 3$ be an odd integer and $n=d^2$. Following Ref.~\cite{ouyang2014permutation}, consider
	a PI code encoding $k=1$ logical qubit into $n$ physical qubits such that the code space has an orthonormal basis 
	\be
	|\psi_{\pm}\ra = \frac1{\sqrt{2^d}} \sum_{\ell=0}^d (\pm 1)^\ell \sqrt{{d \choose \ell}} |D^n_{d\ell}\ra,
	\ee
	where $|D^n_m\ra$ is the $n$-qubit Dicke state with $m$ particles,
	\be
	|D^n_m\ra = \frac1{\sqrt{n \choose m}} \sum_{\substack{x\in \{0,1\}^n \\ |x|=m\\}} \; |x\ra.
	\ee
	Theorem 4 of~\cite{ouyang2014permutation} proves that 
	$\calC=\mathrm{span}(\psi_+,\psi_-)$
	is a distance-$d$ code. 
	Simple algebra gives
	\begin{equation}
		\la +^{\otimes n} |\psi_+\ra = \frac1{\sqrt{2^{d+n}}} \sum_{\ell=0}^d \sqrt{
			{d \choose \ell} {n \choose d\ell}}.
	\end{equation}
	The sum over $\ell$ is lower bounded by a single
	term with $\ell=(d-1)/2$. 
	It is well known that 
	\begin{equation}
		{n \choose m} \ge 2^{nH(m/n)} \sqrt{\frac{n}{8m(n-m)}}.
	\end{equation}
	for all  $n\ge 2$ and $m\in [1,n-1]$.
	Noting that $H(1/2-x)\ge 1-O(x^2)$ one gets
	\begin{equation}
		{d \choose (d-1)/2} \ge \Omega\Big(\frac{2^d}{\sqrt{d}}\Big)
	\end{equation}
	and
	\begin{equation}
		{n \choose d(d-1)/2} \ge \Omega\Big(\frac{2^n}{d}\Big).
	\end{equation}
	It follows that 
	$\la +^{\otimes n} |\psi_+\ra  \ge \Omega( d^{-3/4})$ and thus $E_0(\psi_+)\le -\log_2{|\la +^{\otimes n}|\psi_+\ra|^2}=O(\ln{d})$.

\end{document}